\documentclass[submission,copyright,creativecommons]{eptcs}
\providecommand{\event}{FROM 2019} 
\usepackage{underscore}           
\usepackage[utf8]{inputenc}
\usepackage{amsmath}
\usepackage{amsfonts}
\usepackage{amssymb}
\usepackage{amsthm}
\usepackage{mathtools}
\usepackage{listings}
\usepackage{verbatim}
\usepackage{qtree}
\usepackage{xcolor}
\usepackage{graphicx}
\usepackage{tikz}
\usepackage{tikz-cd}
\usetikzlibrary{arrows,automata, matrix, calc}
\usepackage{hyperref}
\usepackage{longtable}
\usepackage{multicol}
\usepackage{chngcntr}
\usepackage{multicol}
\usepackage{pgfplots}
\pgfplotsset{compat=1.11}

\newcommand*\circled[1]{\tikz[baseline=(char.base)]{
            \node[shape=circle,draw,inner sep=2pt] (char) {#1};}}

\theoremstyle{plain}
\newtheorem{theorem}{Theorem}[]
\newtheorem{lemma}[theorem]{Lemma}

\theoremstyle{definition}
\newtheorem{definition}[theorem]{Definition}
\newtheorem{example}[theorem]{Example}
\newtheorem{problem}[theorem]{Problem}
\newtheorem{algorithm}[theorem]{Algorithm}
\newtheorem{proposition}[theorem]{Proposition}
\newtheorem{question}[theorem]{Question}

\newtheoremstyle{named}{}{}{}{}{}{.}{.5em}{\thmnote{\textit{#1} }#3}
\theoremstyle{named}
\newtheorem*{subalgorithm}{Subalgorithm}

\theoremstyle{remark}

\author{
Georgiana Şurlea 
\institute{Department of Computer Science\\
West University\\
Timişoara, Romania}
\institute{}
\email{\quad \href{mailto:georgiana.surlea96@e-uvt.ro}{\texttt{georgiana.surlea96@e-uvt.ro}}}
\and 
Adrian Crăciun\thanks{Supported by ATCO, “Advanced techniques in combinatorial optimization and computational complexity”,
CNCS IDEI Grant PN-III-P4-ID-PCE-2016-0842.}
\institute{Department of Computer Science\\
West University\\
Timişoara, Romania}
\institute{Institute E-Austria Timişoara}
\email{\qquad \href{mailto:adrian.craciun@e-uvt.ro}{\texttt{adrian.craciun@e-uvt.ro}}}
}
\def\titlerunning{Gr\"obner Bases with Reduction Machines}
\def\authorrunning{G. Şurlea \& A. Crăciun}

\title{Gr\"obner Bases with Reduction Machines}

\begin{document}
\maketitle

\begin{abstract}
In this paper, we make a contribution to the computation of Gr\"obner bases. For polynomial reduction, instead of choosing the leading monomial of a polynomial as the monomial with respect to which the reduction process is carried out, we investigate what happens if we make that choice arbitrarily. It turns out not only this is possible (the fact that this produces a normal form being already known in the literature), but, for a fixed choice of reductors, the obtained normal form is the same no matter the order in which we reduce the monomials. 

To prove this, we introduce reduction machines, which work by reducing each monomial independently and then collecting the result. We show that such a machine can simulate any such reduction. We then discuss different implementations of these machines. Some of these implementations address inherent inefficiencies in reduction machines (repeating the same computations). We describe a first implementation and look at some experimental results. 
\end{abstract}

\section{Introduction}
The concept of Gr\"obner bases, together with an algorithm for the computation of Gr\"obner bases, introduced by Buchberger in \cite{buchberger1965algorithm}, represents an important contribution to symbolic computation. There are many applications of Gr\"obner bases computations, e.g. solving systems of polynomial equations, theorem proving in geometry, software and hardware verification, robotics, coding theory, oil extraction, etc. Descriptions of the applications of Gr\"obner bases can be found for example in~\cite{buchberger1998},~\cite{BuchKauers}. 

\subsection{Informal Description of Gr\"obner Bases}

We have a \emph{set of objects} (multivariate polynomials over a field), together with a \emph{reduction relation} $\rightarrow_F$, generated by a set $F$ (of "generators" for a polynomial ideal). The reduction relation is terminating. This induces a notion (and algorithm) for computing the \emph{reduced normal form} of a polynomial $s$ w.r.t $F$, $Red(s, F)$. 

We want to solve a (variant of the) \emph{word problem}: decide whether a polynomial belongs to the ideal generated by $F$. 
For this, using Buchberger's algorithm (Algorithm~\ref{alg:BuchbergerInformal}), we can construct $G$, a finite set of polynomials that generates the same ideal as $F$. The set $G$ is called \emph{a Gr\"obner basis} and it makes the ideal membership test easy. 

A Gr\"obner basis $G$ is characterized by the fact that certain \emph{critical polynomials}, built from pairs of polynomials in $G$ -- \emph{the S-polynomials} (computed by a function $Spol$) -- can be reduced to zero, i.e. $Red(Spol(f,g), G) = 0$.   

\begin{algorithm}[Buchberger's algorithm for computing Gr\"obner bases]\textit{ }
\label{alg:BuchbergerInformal}
\begin{lstlisting}[escapeinside={*}{*}]
    Input:  *$F$* a finite set of polynomials. 
    Output: *$G$* a finite *Gr\"obner* basis equivalent to *$F$*
        *$G$ := $F$*
        *$C$ := $G \times G$*
        while *$C \neq 0$* do
            Choose a pair *$(f, g) \in C$*
            *$C$* := *$C \setminus \{(f, g)\}$*
            *$h$* := *$Red(Spol(f,g), G)$*
            if *$h \neq 0$* then 
                *$C$*:=*$C \cup (G\times \{h\})$*
                *$G$*:= *$G \cup \{h\}$*
        return *$G$*
\end{lstlisting}
\end{algorithm}


\subsection{Improvements of the Gr\"obner Bases Algorithm}
\label{GBimprovements}
Due to the high complexity of the problem (double exponential in the worst case, see~\cite{MAYR1982}), there have been (and are continuing) efforts for improving performance. We outline some of them here. 

\emph{Criteria to avoid useless reductions} have been proposed by Buchberger in \cite{Buchberger1979}. \emph{Reduced Gr\"obner bases} were proposed in~\cite{buchberger1985}: polynomials in the basis being constructed (before the start of the algorithm), as well as before adding new polynomials to the basis, should be reduced among each other in order to eliminate redundant polynomials in the basis. 

The algorithm for computing Gr\"obner bases leaves some choices open. \emph{Which critical pair should be selected} -- simple criteria like minimal degree for the least common multiple of the components were proposed in~\cite{buchberger1985}, as well as more sophisticated ones like the "sugar" strategy, see~\cite{Giovini1991}. \emph{Which element from the basis is chosen for the current reduction step} -- a discussion of the possible choices can be found in~\cite{mora2003solving}.  

Further improvements looked at combining techniques from linear algebra to allow multiple reductions of S-polynomials, e.g. the F4~\cite{FAUGEREF4} and F5~\cite{FaugereF5} algorithms, or techniques from partial differential equations together with a restriction of the notion of reduction to construct \emph{involutive bases} containing Gr\"obner bases, see~\cite{GERDT1998519}, also~\cite{APEL20051131}. 

\subsection{Contribution and Organization}

Our focus in this paper is the basic process of reducing a polynomial modulo a set of polynomials, which is the most time consuming part in the execution of the Gr\"obner bases algorithm. Traditionally, for the polynomial being reduced, one selects the leading monomial (w.r.t. some admissible ordering) and reduces it (in a while loop, for each such situation). 

What would happen if we choose a different monomial? Would reduction work? Yes, a normal form would be computed. Moreover, in certain conditions (if we keep the choice of reductors, i.e. the selection strategy), no matter the choice of the monomial, we get \emph{the same} normal form. Would different choices lead to a faster solution?  Probably not, but sometimes it can be done just as fast. However, this line of inquiry led to what turns out to be a parallel method to do reduction, which has the potential to improve the performance. 

Our paper is organized as follows. Section~\ref{preliminaries} provides the basic concepts needed, including standard reduction. In Section~\ref{lazyreduction} we introduce \emph{reduction processes}, objects that express all possible ways in which we can perform reductions.  In Section~\ref{quickReduction} we introduce \emph{reduction machines} and prove that for a fixed choice of reductors these are equivalent to the corresponding reduction process. 
Section~\ref{implementation} discusses algorithms that implement reduction machines as well as our initial implementation of Gr\"obner bases with reduction machines and some experiments on the impact of reduction machines on the computation of Gr\"obner bases. Section~\ref{related} discusses related work, while Section~\ref{conclusions} concludes with a discussion of open problems and future research directions.



\section{Preliminaries}
\label{preliminaries}

Since the focus of this paper is reduction, we recall some relevant notions. For this we use the notation and definitions from~\cite{buchberger1998}. For the complete details, we refer the reader to standard textbook presentations of Gr\"obner bases, e.g.~\cite{buchberger1998},~\cite{becker1993grobner}, etc.

\begin{longtable}{rl}
$\left(\mathbb{K}, +, 0, -, \cdot, 1, /\right)$ & a \textit{field}, \\
$\mathbb{K}[x_{1}, \ldots, x_{n}]$ & the ring of $n$-variate polynomials over $\mathbb{K}$, \\
$[x_1, \ldots, x_n]$ & power products over $x_1, \ldots, x_n$,\\
$f, g, h, p, q, s$ & polynomials in $\mathbb{K}[x_{1}, \ldots, x_{n}]$,\\
$F, G$ & finite subsets of $\mathbb{K}[x_{1}, \ldots, x_{n}]$, \\
$m$ & monomials, \\
$t, u, pp$ & power products from $[x_1, \ldots, x_n]$, \\
$a, b, c, d$ & elements in $\mathbb{K}$, \\
$i, j, k, l, n, o, q$ & natural numbers, \\
$t|u$& $t$ divides $u$,\\
$t/u$& division of power products (in case $t|u$), \\
$C(m)$ & the coefficient of monomial $m$, \\
$C(p, t)$ & the coefficient at $t$ in the polynomial $p$,\\
$M(p, t) = C(t, p)\cdot t$ & the monomial at $t$ in $p$, \\
$S(p) = \left\{t | C(p, t)\neq 0 \right\}$ & the support of polynomial $p$.\\
\end{longtable}

Moreover, we will use the following shortcut notation: $$\mathbf{P}:=\mathbb{K}[x_1, \ldots, x_n], \mathbf{T}:=[x_1, \ldots, x_n].$$ 

\begin{definition}[Admissible ordering]
Let $\prec$ be a total ordering on $\mathbf{T}$. Then, 
$$\prec \text{ is admissible }:\Leftrightarrow \begin{array}{l}
\forall t\neq 1 (1\prec t), \\
\forall t, u, v (t\prec u \Rightarrow t\cdot v \prec u \cdot v). 
\end{array}
$$

\end{definition}
Admissible orderings include the lexicographic ordering and reverse lexicographic ordering, total degree lexicographic and reverse lexicographic. See, for example,~\cite{becker1993grobner} for technical details.  
\begin{proposition}[Properties of admissible orderings]

Let $\prec$ be an admissible ordering on $\mathbf{T}$. Then, 
$$
\begin{array}{l}
\forall t, u (t| u \Rightarrow t \preceq u), \\
\prec \text{ is well founded}. 
\end{array}
$$
\end{proposition}

Admissible orderings allow decompositions of polynomials:

\begin{longtable}{rl}
$LPP_{\prec}(p)=max_{\prec}(S(p))$ & the leading power product of $p$, \\
$LC_{\prec}(p)=C(p, LPP_{\prec}(p))$& the leading coefficient of $p$, \\
$LM_{\prec}(p)=LC_{\prec}(p)\cdot LPP_{\prec}(p)$ & the leading monomial of $p$,\\
$R_{\prec}(p)= p - LM_{\prec}(p)$ & the remaining part of $p$. \\

\end{longtable}
Note that in case the admissible ordering is fixed, we will not write it as a subscript (e.g. will write $LPP(p)$ instead of $LPP_{\prec}(p)$).  






\begin{definition}[Reduction modulo polynomials]\textit{ }

\emph{$g$ reduces to $h$ modulo $f$ using the power product $t$}, 
$$g\rightarrow_{f,t} h :\Leftrightarrow t \in S(g) \wedge LPP(f)|t \wedge h = g -(M(g, t)/LM(f))\cdot f. $$

\emph{$g$ reduces to $h$ modulo $f$}, 
$$g\rightarrow_{f} h :\Leftrightarrow \exists t \in S(g) (g\rightarrow_{f, t}h).$$

\emph{$g$ reduces to $h$ modulo the set $F$ of polynomials}. 
$$ g\rightarrow_{F} h:\Leftrightarrow \exists f \in F (g\rightarrow_{f} h).$$

\end{definition}

\begin{definition}[Normal forms]\textit{ }

\emph{$g$ is in normal form modulo $F$}, 
$$\underline{g}_{F}:\Leftrightarrow \nexists h (g\rightarrow_{F}h).$$

Let ${\rightarrow_{F}}^*$ be the reflexive transitive closure of $\rightarrow_{F}$. \emph{$h$ is a normal form of $g$ modulo $F$} iff 

$$g{\rightarrow_{F}}^{*}h \wedge \underline{h}_{F}.$$

\end{definition}

The reduction process can be iterated algorithmically.

\begin{definition}[Normal form algorithm]
An algorithm $S$ is called \emph{a normal form algorithm} (or \emph{simplifier}) iff
$$g{\rightarrow_{F}}^{*}S(F, g) \wedge \underline{S(F,g)}_{F}.$$

\end{definition}


Some textbooks, e.g.~\cite{cox}, propose a maximal normal form algorithm, where the reduction is done on the leading monomial.  We will refer to it in this paper as the \emph{classic reduction algorithm}.

\begin{lemma}
The classic reduction algorithm (Algorithm~\ref{standardReduction}) is a normal form algorithm. 
\begin{algorithm}[Classic reduction]\textit{ }
\label{standardReduction}
\begin{lstlisting}[escapeinside={*}{*}]
    *$h := g$*
    while exists *$f \in F$* such that *$LPP(f)|LPP(h)$* do
        choose *$f \in F$* such that *$h \rightarrow_{f, LPP(h)}$* 
        *$h := h - \dfrac{1}{LC(h)} \cdot \dfrac{LPP(h)}{LPP(f)} \cdot f$*
    return h
\end{lstlisting}
\end{algorithm}
\end{lemma}

The correctness and termination of the algorithm can be found in \cite{becker1993grobner}.

\begin{example}
Consider $F = \{ f_1, f_2 \}$, where $f_1 := x^2 + x -y,$ $f_2 := x - 2.$ The polynomials $f_1, f_2$ are ordered according to the degree lexicographic ordering. The leading power products are $x^2$, $x$, respectively, and the leading coefficients are 1 and 1. 

Consider $g := x^3 + x^2y + 2y$. After one step of reduction, $g$ reduces modulo $F$ to
$$h := x^2y - x^2 +xy + 2y.$$
Namely, 

$$h := g - 1 \cdot x \cdot f_1.$$
The classic reduction of $g$ modulo $F$ yields the following sequence:

\begin{tabular}{l l}
$g = x^3 + x^2y + 2y$ & (using $f_1$) \\
$\rightarrow_{F} x^2y - x^2 + xy + 2y$ & (using $f_1$) \\
$\rightarrow_{F} - x^2 + y^2 + 2y$ & (using $f_1$) \\
$\rightarrow_{F} y^2 + x + y$ & (using $f_2$) \\
$\rightarrow_{F} y^2 +y + 2$. & \\
\end{tabular}
\end{example}

\section{Reduction Processes}
\label{lazyreduction}

Now, we consider what happens if, for the polynomial being reduced, we allow the choice of any monomial, not just the maximal one, as described in Algorithm~\ref{standardReduction}. This is can be done, literature indicates the maximal choice is made for efficiency purposes (see, for example~\cite{buchberger1985}). Normal forms can be computed using other choices of monomials.  








\begin{definition}[Reduction process]
\label{ReductionProcesses}

Let $g = m_{1} + \dots + m_{n}$ and let $F = \{ f_{1}, \dots, f_{l} \}$, such that a \emph{selection strategy} for choosing elements from $F$ to use for reduction is fixed. A \emph{reduction process} of $g$ modulo $F$ represents all possible reductions of $g$ modulo $F$ to normal form.


\end{definition}

\begin{definition}[Monomial reduction sequence]
A \emph{monomial reduction sequence} is the sequence of monomials that were selected for reduction in the computation of a normal form  within a reduction process.
\end{definition}

The systematic enumeration of all reductions yields a tree, similar to that in Figure~\ref{fig:tree}.
A monomial reduction sequence implicitly describes a branch of that tree. 

\begin{figure}[ht!]
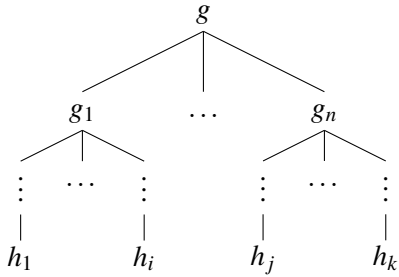

\Tree[.$g$ [.$g_{1}$ [.$\vdots$ [ .$h_{1}$ ]
                               ]
						       [.$\dots$ 
                               ]
						       [.$\vdots$ [ .$h_{i}$ ]
                               ]
				]
				[.$\ldots$ 
                ]
         	    [.$g_n$ [.$\vdots$ [ .$h_{j}$ ]
                               ]
						       [.$\dots$ 
                               ]
               			       [.$\vdots$ [ .$h_{k}$ ]
                               ]
                ]
     ]
\caption{Tree representation of a reduction process.}
\label{fig:tree}
\end{figure}

Now, for a reduction process, we consider the following questions: 
\begin{question}
     Does this process terminate? Yes, see~\cite{becker1993grobner} for the termination proof.
\end{question}
\begin{question}
Does this process lead to a normal form? Yes, also in~\cite{becker1993grobner}. 
\end{question}
\begin{question}
\label{question:uniqueNormalForm}
    Does a reduction process compute one normal form or more?
\end{question}
\begin{question}
\label{question:howQuickIsReduction}
Can we identify the most efficient branch in the reduction process? Is this the standard reduction from Algorithm~\ref{standardReduction}? Are there better choices? Is it easy to decide which choices would lead to short(er) reductions? 
\end{question}
\begin{example}
\label{eg:2}
Consider $F = \{ f_1, f_2 \}$, where 
$f_1 := x^2 + x -y,$
$f_2 := x - 2.$
The polynomials $f_1, f_2$ are ordered according to the degree lexicographic ordering. The leading power products are $x^2$, $x$, respectively, and the leading coefficients are $1$ and $1$. Consider 
$$g := x^3 + x^2y + 2y.$$

The (incomplete) reduction process is illustrated in Figure \ref{fig:ExampleTree}. Each node contains a polynomial. The \textbf{bold face} monomials are reducible and the edges correspond to the order in which monomials are selected for reduction, i.e. first edge corresponds to the reduction using the first monomial, and so on. Monomials that are {\color{gray}grayed out} are reducible, but for the sake of space we left out the corresponding branches (branches that are not shown lead to the same result). The rest of monomials are irreducible.

\begin{figure}[htbp]
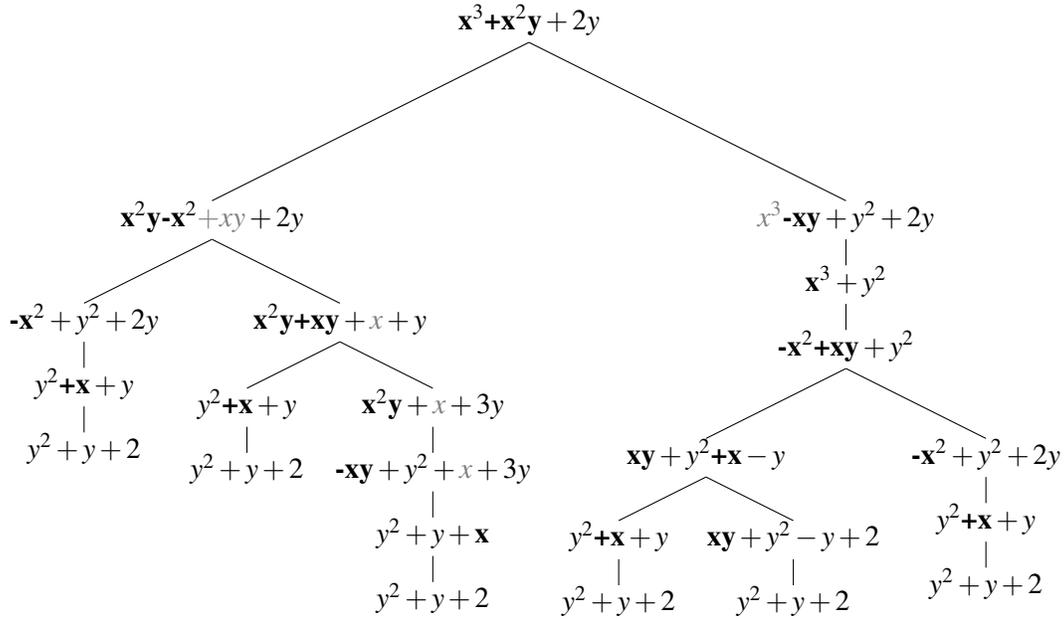

\Tree[.$\textbf{x^3}\textbf{+x^2y}+2y$ [.$\textbf{x^2y}\textbf{-x^2}\textcolor{gray}{+xy}+2y$ [.$\textbf{-x^2}+y^2+2y$ [ .$y^2\textbf{+x}+y$ [ .$y^2+y+2$ ] ]
                               ]
						       [.$\textbf{x^2y}\textbf{+xy}+\textcolor{gray}{x}+y$ [ .$y^2\textbf{+x}+y$ [ .$y^2+y+2$ ] ]
						                       [ .$\textbf{x^2y}+\textcolor{gray}{x}+3y$ [ .$\textbf{-xy}+y^2+\textcolor{gray}{x}+3y$ [ .$y^2+y+\textbf{x}$ [ .$y^2+y+2$ ]                                  ]
						                                      ] 
						                       ]
						      ]
				]
				[.$\textcolor{gray}{x^3}\textbf{-xy}+y^2+2y$
                                                [.$\textbf{x^3}+y^2$ [.$\textbf{-x^2}\textbf{+xy}+y^2$ [.$\textbf{xy}+y^2\textbf{+x}-y$ [.$y^2\textbf{+x}+y$ [.$y^2+y+2$ ] ]
                                                                                           [.$\textbf{xy}+y^2-y+2$ [.$y^2+y+2$ ] ]
                                                                            ]
                                                                            [.$\textbf{-x^2}+y^2+2y$ [.$y^2\textbf{+x}+y$ [.$y^2+y+2$ ] ] ] 
                                                            ] 
                                                ]
				              ]
     ]
     
\caption{Tree representation of the reduction process for Example~\ref{eg:2}.}
\label{fig:ExampleTree}
\end{figure}
\end{example}

Note that we obtain the same normal form on each of the branches. The leftmost branch corresponds to the classic reduction. Compared to the classic reduction, there is one other branch of the same length. The rest correspond to longer reduction chains.

\section{Reduction Machines}
\label{quickReduction}

To answer Question~\ref{question:uniqueNormalForm}, we will first introduce a few notions, then prove that the answer is Yes. 

\begin{definition}[Monomial substitution]
Let $m, m_1, \ldots, m_n$ be monomials, $f$ a polynomial and $F$ a set of polynomials, respectively, such that $m\rightarrow_F f$ and $f = m_1 + \ldots +m_n$. Then the sequence $m_1, \ldots, m_n$  is a \emph{monomial substitution} for $m$ modulo $F$. 
\end{definition}

\begin{definition}[Reduction thread]
Let $m$ be a monomial, $F$ a set of polynomials (with a fixed selection strategy). A \emph{reduction thread} for $m$ modulo $F$ replaces $m$ with its monomial substitution modulo $F$. The process is repeated for as long as there are reducible monomials. 
\end{definition}

Note that a reduction thread can be represented by a tree containing irreducible monomials as leaves. 

\begin{definition}[Reduction machine]
Let $g = m_1+ \ldots + m_n$ be a polynomial represented as the sum of monomials and $F$ a set of polynomials. A \emph{reduction machine} with inputs $g$ and $F$ is described in the following way:  
\begin{itemize}
    \item for each of the monomials $m_i, i=1, \ldots, n$, construct its reduction thread, 
    \item for the resulting sequence of reduction threads, accumulate the sum of all the leaves (irreducible monomials) and return the result. 
\end{itemize}
\end{definition}

Note that the reduction threads are independent of each other, by construction, so in principle they can be executed in the same time.

\begin{definition}[Execution trace]
An \emph{execution trace} of a reduction machine represents a configuration of the reduction machine, i.e. the reduction machine where the reduction threads have not been completed. This means that leaves of this execution trace are not necessarily irreducible. 
\end{definition}

\begin{example}
\label{example:comparison}
Let $g = 4x^3+2x^2y+7xy+2y$ and $F = \{x^2+x-y, x-2\}$ for which, as selection strategy, we choose polynomials according to the lexicographic ordering of their leading monomials. The classic reduction yields a sequence of the following form:

$$    \begin{array}{ll}
       \text{}                                         & 4x^3+2x^2y+7xy+2y \\ 
    \xrightarrow{\enspace\circled{1}\enspace} & 2x^2y-4x^2+11xy+2y \\ 
    \xrightarrow{\enspace\circled{2}\enspace} & -4x^2+9xy+2y^2+2y \\ 
    \xrightarrow{\enspace\circled{3}\enspace} & 9xy+2y^2+4x-2y \\ 
    \xrightarrow{\enspace\circled{4}\enspace} & 2y^2+4x+16y, \\ 
    \xrightarrow{\enspace\circled{5}\enspace} & 2y^2+16y+8
\end{array}$$

The corresponding reduction machine is represented in Figure~\ref{reductionMachine}.
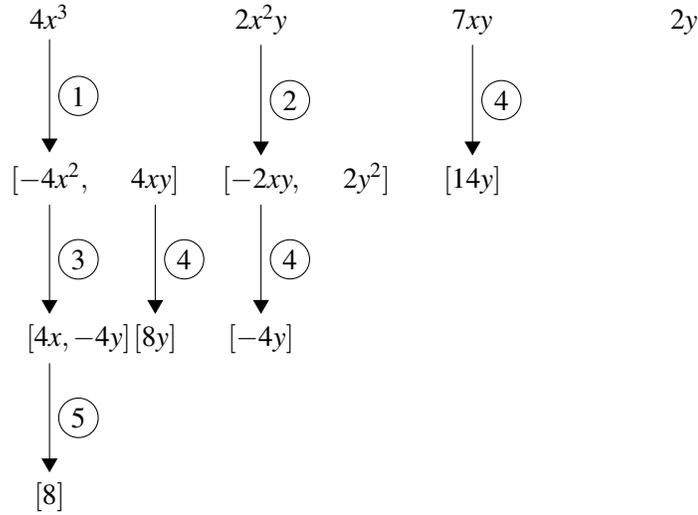
\begin{figure}[ht!]
\begin{center}
\begin{tikzpicture}[>=triangle 60]
  \matrix[matrix of math nodes,column sep={40pt,between origins},row
    sep={60pt,between origins},nodes={asymmetrical rectangle}] (s)
  {
    |[name=k1]| 4x^3 & |[name=k2]| & |[name=k3]| 2x^2y & &|[name=k4]| 7xy & & 2y \\
    |[name=A11]| [-4x^2, &|[name=A13]| 4xy] & |[name=A21]| [-2xy, & 2y^2] &|[name=A31]| [14y] \\
    |[name=B11]| \qquad [4x, -4y] &|[name=B12]| [8y] & |[name=B21]| [-4y]  \\
    |[name=C11]| [8] & & |[name=C21]|  \\
  };
  \draw[->] (k1) edge node{\qquad\circled{1}} (A11) 
            (k3) edge node{\qquad\circled{2}} (A21)
            (k4) edge node{\qquad\circled{4}} (A31)
            (A11) edge node{\qquad\circled{3}} (B11)
            (A13) edge node{\qquad\circled{4}} (B12)
            (A21) edge node{\qquad\circled{4}} (B21)
            (B11) edge node{\qquad\circled{5}} (C11)
  ;
\end{tikzpicture}
\end{center}
\caption{Reduction machine for Example~\ref{example:comparison}.}
\label{reductionMachine}
\end{figure}


Note that the reduction machine computes the same result as the standard reduction. It can simulate the steps of the standard reduction in a manner that we explain below, in the proof of Theorem~\ref{theorem:sameRemainder}. The numbers attached to the reduction arrows show the correspondence between the reduction machine and classic reduction. If the reduction threads are running in the same time, the reduction machine can compute the results faster than the classic reduction: classic reduction has depth 5, whereas the reduction machine has depth 3. However, if the reduction threads are running sequentially, the depth of the reduction machine becomes 7. 
\end{example}

\begin{theorem}
\label{theorem:sameRemainder}
Let $g = m_{1} + \dots + m_{n}$  let $F = \{ f_{1}, \dots, f_{l} \}$ be an ordered $l$-tuple of polynomials with a fixed selection strategy. Each branch of the reduction process of $g$ with respect to $F$ yields the same result.  
\end{theorem}

\begin{proof}
\label{proof:substitutions}

To prove the theorem, we will prove that the reduction process of $g$ with respect to $F$ and the associated reduction machine are equivalent. This means that since the reduction machine yields a unique result, so does the reduction process. 

We prove that any node $p$ in the reduction process is computed by one execution trace in the corresponding reduction machine. Let $s$ be the (partial) monomial reduction sequence that defines the path from the root of the reduction process to $p$.

We prove this by induction on the length of the monomial reduction sequence. 

\noindent \textbf{Base case}: 

In this case, $p$ is the root of the reduction process, therefore, by definition, the polynomial is represented in the reduction machine. 

\noindent \textbf{Induction step}: 

Assume that the property holds for monomial reduction sequence $S$, i.e. for polynomial $p$ defined by $S$ we have a unique corresponding execution trace in the reduction machine.  

Let $m=c\cdot t$ be a monomial chosen from $p$, where $c$ is its coefficient, $t$ its power product. We show the property is true for the monomial reduction sequence $S\cup \{m\}$. 

We write $p=m + R(p)$. Let $ET$ be the execution trace of the reduction machine that corresponds to polynomial $p$. Since $m$ is a monomial in $p$, it is collected from leaves of (partial) reduction threads from $ET$. Let these leaves be $c_1\cdot t, \ldots, c_n\cdot t$, $n\geq 1$, we know $c = c_1+\ldots c_n$. 

Let $p^{\prime}$ be the polynomial such that $p\rightarrow_{f, t} p^{\prime}$  by choosing the monomial $m$, using $f\in F$ defined by the fixed selection strategy. 

We show that there exists a unique execution trace $ET^{\prime}$ that corresponds to $p^{\prime}$. By the definition of reduction $$p^{\prime} = p - \frac{m}{LM(f)} \cdot R(f).$$ In other words, the monomial $m$ is substituted by the polynomial $$s = - \frac{m}{LM(f)} \cdot R(f),$$ i.e.   $$s = c\cdot \left(-\frac{t}{LM(f)} \cdot R(f)\right).$$

Now, let $ms_1, \dots, ms_k$ be the monomials of the polynomial $s$. We build $ET^{\prime}$ from $ET$ in the following way: we extend the reduction threads containing $c_1\cdot t, \ldots, c_n\cdot t$ as leaves by adding the corresponding substitution of $t$, obtaining new leaves $c_i \cdot ms_1, \dots, c_i \cdot ms_k$, $i = 1, \dots, n$. 

Collecting the leaves of $ET^{\prime}$ yields $p^{\prime}$. 

To summarize, we have shown that between making a step in the reduction process and extending an execution trace using a monomial there is a one-to-one correspondence, which concludes the proof.

\end{proof}

Theorem~\ref{theorem:sameRemainder} solves Question~\ref{question:uniqueNormalForm}. This establishes confluence for \emph{reduction with a fixed reductor selection strategy}, i.e. the choice of monomials that are to be reduced does not influence the final result of the reduction. This is different from the notion of confluence that occurs in the context of Gr\"obner bases in the following sense: for polynomial set $F$,  $\rightarrow_F$ is not confluent in general, i.e. different choices of reductors from $F$ may lead to different normal forms. However, if $F$ is a Gr\"obner basis, then $\rightarrow_F$ is confluent.  

For Question~\ref{question:howQuickIsReduction}, we initially considered trying to determine criteria that would make a computation short, solving an optimization problem expressed in terms of degrees. However, the equivalence between reduction processes and reduction machines allows us to use the latter for doing reductions in parallel. We propose reduction machines as a candidate answer to Question~\ref{question:howQuickIsReduction}.

\section{Implementation and Experiments}
\label{implementation}
Here, we first present an algorithm that implements a reduction machine, then discuss some of its limitations and show how to improve it. We then briefly describe how we integrated these algorithms into a Java implementation of Gr\"obner bases. Finally, we discuss some experimental results. 

\subsection{Versions of Reduction Machines}
The following algorithm is a straight-forward implementation of a reduction machine.

\begin{algorithm}[Reduction machine]\textit{ }
\label{alg:quickReduction}
\begin{lstlisting}[escapeinside={*}{*}]
    *$M$ := monomials of $g$*
    *$h$ := 0*
    
    while *$M$ is not empty* do
        *$m := \text{a monomial in } M$*
        if *$m$ is reducible* then
            *$S$ := substitution($m$)*
            *$M := M \cup S$*
        else
            *$h := h + m$*
        *$M := M \setminus \{m\}$*
    return *$h$*
\end{lstlisting}
\end{algorithm}


Note, however, that this implementation has some inefficiencies. These are visible in Figure~\ref{reductionMachine}: the power product $xy$ appears in 3 separate reduction threads, therefore we need to perform the same reduction 3 times. Furthermore, there could be situations where we reduce the same power product several times, only for this to be cancelled out after collecting the result. 

We can optimize the implementation by introducing a caching mechanism that  will allow us to detect reductions that were already carried out and only update the coefficients. For this, we use the following structures: 
\begin{itemize}
    \item A \emph{sequence of monomials} with two types of markings. The marking $unused$ represents a monomial which has not yet been processed. The marking $og$ indicates a monomial of $g$. 
    \item A \emph{directed graph} to represent reduction threads. A vertex in this graph consists of a power product and a set of multiples. These multiples contain coefficients of the power products from different reduction threads, as well as coefficients inherited (by reduction) from parent vertices (if applicable). Edges in the graph show reduction steps. 
\end{itemize}




We now give the details.

\begin{algorithm}[Reduction machine with caching]\textit{}
\label{alg:quickReductionCaching}
\begin{lstlisting}[escapeinside={*}{*}]
    *$M$ :=  monomials of $g$*  (*$unused$, $og$*)
    *$G$* := empty graph
    
    while *exists $unused$ in $M$* do
        *$m := \text{an $unused$ monomial in } M$*
        if *$m$ is $og$* then 
            if existsVertex(*$pp(m)$, $G$*) then
                updateVertex(*$pp(m)$, $c(m)$, $G$*)
            else 
                createVertex(*$pp(m)$, $c(m)$, $G$*)
        if *$m$ is reducible* then
            *$S$ := substitution($pp(m)$)*
            expand(*$pp(m)$, $S$, $G$*)
            update(*$M$, $S$*)
        mark(*$m$, used*)
    return collectRemainder(*$G$*)
\end{lstlisting}
\end{algorithm}

\begin{subalgorithm} [expand ($pp$, $S$, $G$)]\textit{}
\label{subalg:expand}
\begin{lstlisting}[escapeinside={*}{*}]
    *$source$* := getVertex(*$pp$, $G$*)
    mark(*$source$*, reducible)
    for each *$m$ in $S$* do 
        if containsVertex(*$pp(m)$, $G$*) then
            *$destination$ := getVertex($pp(m)$, $G$)*
            addMultiple(*$c(m)$, $source$, $destination$*)
        else
            createVertex(*$pp(m)$, $c(m)$, $source$*)
            *$destination$ := getVertex($pp(m)$, $G$)*
        addEdge(*$source$, $destination$, $G$*)
\end{lstlisting}
\end{subalgorithm}

\begin{subalgorithm}[update ($M$, $S$)]\textit{}
\label{subalg:update}
\begin{lstlisting}[escapeinside={*}{*}]
    for each *$m$ in $S$* do
        if *$pp(m)$ does not exist in $M$* then
            *$M$ := $M \cup \{m\}$*
\end{lstlisting}
\end{subalgorithm}

\begin{subalgorithm}[collectRemainder ($G$)]\textit{}
\label{subalg:create}
\begin{lstlisting}[escapeinside={*}{*}]
    *$h$ := 0*
    *$I$ := irreducibleVertices($G$)*
    for each *$v$ in $I$* do
        *$c$ = collectCoefficients($v$, $G$)*
        if *$c \neq 0$* then
            *$h$ := $h$ + $c \cdot \text{getPowerProduct(}v$)*
    return *$h$*
\end{lstlisting}
\end{subalgorithm}

\begin{subalgorithm}[collectCoefficients ($v$, $G$)]\textit{}
\label{subalg:create1}
\begin{lstlisting}[escapeinside={*}{*}]
    *$s$ := 0*
    *$C$ := getMultiples($v$, $G$)*
    for each *$m$ in $C$* do
        *$s$ := $s$ + getCoefficient($m$)*
        if hasParent(*$m$*) then
            *$s$ := $s \cdot \text{collectCoefficients(getParent(}m$))*
    return *$s$*
\end{lstlisting}
\end{subalgorithm}


The \verb|collectRemainder(G)| method starts from the vertices containing irreducible power products, i.e., leaf or isolated vertices, and propagates backwards the multiples across vertices containing reducible power products. An irreducible vertex inherits the coefficients of all reducible vertices connected to it.

\begin{example}
Let $F = \{x^2+x-y, x-2\}$ and $g = 4x^3+2x^2y+7xy+2y$ be an ideal of polynomials and a polynomial, respectively, ordered with respect to the lexicographic ordering. The oriented graph constructed for this example is illustrated in Figure~\ref{reductionGraph}.

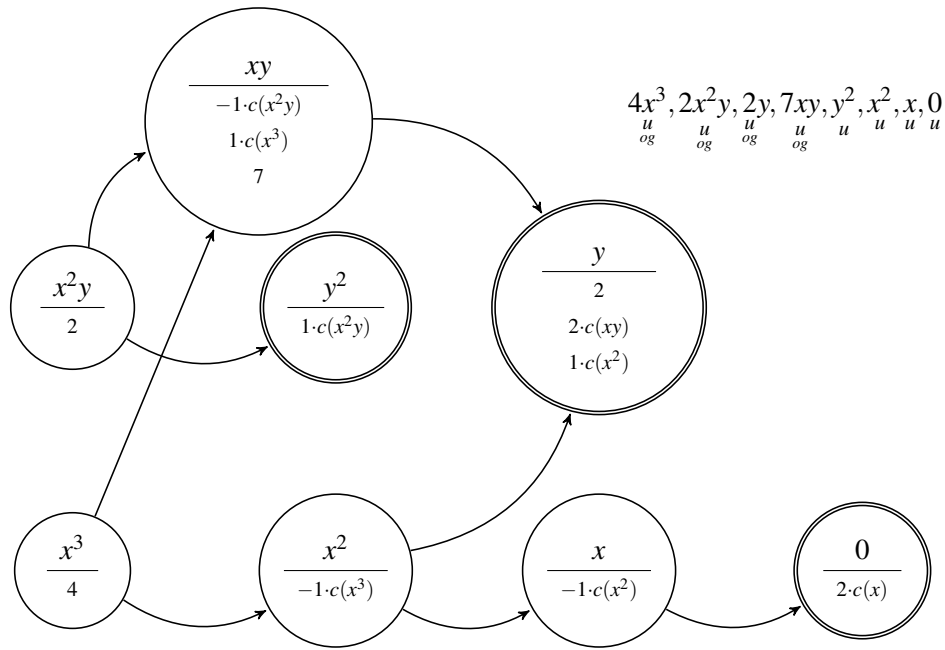
\begin{figure}[ht!]
\begin{center}
\begin{tikzpicture}[->,>=stealth',shorten >=1pt,auto,node distance=3.5cm,
                    semithick]
  \tikzstyle{every state}=[]

  \node[state]                  (A)                    {$\begin{array}{c}x^2 y \\\hline{\scriptstyle 2}\end{array}$};
  \node[state]                  (B) [below of=A]       {$\begin{array}{c}x^3\\\hline{\scriptstyle 4}\end{array}$};
  \node[state]                  (C) [above right of=A] {$\begin{array}{c}xy\\\hline \begin{array}{c}{\scriptstyle-1\cdot c(x^2y)}\\{\scriptstyle 1\cdot c(x^3)}\\ {\scriptstyle7}\end{array}\end{array}$};
  \node[state, accepting]       (D) [right of=A] {$\begin{array}{c}y^2\\\hline{\scriptstyle 1 \cdot c(x^2y)}\end{array}$};
  \node[state]                  (E) [right of=B]       {$\begin{array}{c}x^2\\\hline{\scriptstyle-1\cdot c(x^3)} \end{array}$};
  \node[state, accepting]        (F) [right of=D] {$\begin{array}{c}y\\\hline \begin{array}{c} {\scriptstyle 2}\\ {\scriptstyle 2\cdot c(xy)}\\{\scriptstyle1\cdot c(x^2)}\end{array}\end{array}$};
  \node[state]                  (G) [right of=E]       {$\begin{array}{c}x\\\hline {\scriptstyle -1 \cdot c(x^2)}\end{array}$};
  \node[state, accepting]       (H) [right of=G] {$\begin{array}{c}0\\ \hline{\scriptstyle 2 \cdot c(x)}\end{array}$};
  \node[]                   (I) [above right of=F] {$\underset{\underset{og}{u}}{4x^3}, \underset{\underset{og}{u}}{2x^2 y}, \underset{\underset{og}{u}}{2}y, \underset{\underset{og}{u}}{7xy}, \underset{u}{y^2}, \underset{u}{x^2}, \underset{u}{x}, 
  \underset{u}{0}$};

  \path (A) edge [bend left]             node {} (C)
            edge [bend right]             node {} (D)
        (B) edge  [bend right]            node {} (E)
            edge [] node {} (C)
        (C) edge [bend left]  node {} (F)
        (E) edge [bend right] node {} (F)
            edge [bend right] node {} (G)
        (G) edge [bend right]  node {} (H);
\end{tikzpicture}
\caption{Reduction machine with caching, resulting graph and associated monomial sequence.}
\label{reductionGraph}
\end{center}
\end{figure}
\end{example}

\subsection{A Java Library for Gröbner Bases}
Our prototype Java library for Gröbner bases provides implementations for the following concepts, algorithms and problems: defining multivariate polynomials with coefficients over the field of rational numbers ($\mathbb{Q}$), defining ideals, reducing a polynomial with respect to a basis, deciding whether a basis is a Gröbner basis, computing a Gröbner basis given an ideal using Buchberger's standard and improved algorithms, testing whether two given polynomials are congruent with respect to a basis, deciding whether a polynomial belongs to the ideal generated by a Gröbner basis. The available orderings are declared in the public enumeration \verb|Ordering|, as follows: lexicographic (LEX), reverse lexicographic (REVLEX), graded lexicographic (GRLEX), reverse graded lexicographic (GREVLEX). The versions of Buchberger's algorithm are declared in the public enumeration \verb|GroebnerType|, as follows: standard version (CLASSIC) and improved version (IMPROVED).

\subsection{Experimental Results}
We integrated sequential versions of the reduction machine in our implementation. We ran a collection of 20 problems selected from literature. In order to account for any discrepancy between the execution time values of two successive runs, we computed the execution time as the average time for 1000 runs. 

We compared the classic reduction and both versions of the reduction machine, integrated in the improved version of Buchberger's algorithm for computing reduced bases. We compare the results by providing the number of the problem, without its structure. For the structure of the problem, we refer the reader to the collection listed in the Appendix. The resulting bases were computed with respect to the graded lexicographic ordering (GRLEX). Results are shown in Figure~\ref{chart:results}. C indicates the use of classic reduction, RM the reduction machine and RMc the reduction machine with caching. 
\pgfplotstableread[row sep=\\,col sep=&]{
    interval & C & RM & RMc \\
    1     & 234.227731  & 119.622076 & 136.947616 \\
    2     & 101.27696 & 105.201283 & 100.431257 \\
    3     & 89.133154 & 58.919506 & 80.247366 \\
    4     & 98.763084 & 94.514982 & 65.0487 \\
    6     & 26.877505  & 40.394339 & 55.636601 \\
    7     & 20.158077  & 13.392925 & 15.655764 \\
    8     & 22.841591  & 26.108469 & 28.018169 \\
    9     & 12.961462  & 15.153644 & 15.928023 \\
    11    & 76.676436  & 142.942679 & 114.37193 \\ 
    12    & 20.045247  & 33.717867 & 51.739979 \\
    13    & 86.039023  & 61.054938 & 51.184555 \\
    14    & 47.01247  & 58.947765 & 62.741569 \\
    16    & 48.510105  & 53.877727 & 58.890123 \\
    17    & 44.831707  & 96.252473 & 66.989209 \\
    18    & 138.989664  & 152.796006 & 161.905974 \\
    }\RTI
   
\pgfplotstableread[row sep=\\,col sep=&]{
    interval & C & RM & RMc \\
    5     & 529.799822  & 926.3537 & 597.447242 \\
    10    & 303.277355  & 310.425133 & 321.958031 \\
    cyclic-3 & 231.070074 & 141.595445 & 159.200319 \\
    cyclic-4 & 614.188353 & 917.519465 & 711.187494 \\
    }\RTII
    
\begin{figure}[ht!]
\centering
\begin{tikzpicture}
    \begin{axis}[
            ybar,
            bar width=.2cm,
            width=\textwidth,
            height=.4\textwidth,
            legend style={at={(0.86,1)},
                anchor=north,legend columns=-1},
            symbolic x coords={1, 2, 3, 4, 6, 7, 8, 9, 11, 12, 13, 14, 16, 17, 18},
            xtick=data,
            ytick={0, 50, 100, 150, 200, 250},
            nodes near coords align={horizontal},
            ymin=0,ymax=250,
            ylabel={Time (ms)},
        ]

        \addplot table[x=interval,y=C]{\RTI};
        \addplot table[x=interval,y=RM]{\RTI};
        \addplot table[x=interval,y=RMc]{\RTI};
        \legend{C, RM, RMc}
    \end{axis}
\end{tikzpicture}

\begin{tikzpicture}
    \begin{axis}[
            ybar,
            bar width=.2cm,
            width=.6\textwidth,
            height=.4\textwidth,
            legend style={at={(0.86,1)},
                anchor=north,legend columns=-1},
            symbolic x coords={5, 10, cyclic-3, cyclic-4},
            xtick=data,
            nodes near coords align={horizontal},
            ymin=0,ymax=1000,
            ylabel={Time (ms)},
        ]

        \addplot table[x=interval,y=C]{\RTII};
        \addplot table[x=interval,y=RM]{\RTII};
        \addplot table[x=interval,y=RMc]{\RTII};
    \end{axis}
\end{tikzpicture}
\caption{Experimental results.}
\label{chart:results}
\end{figure}
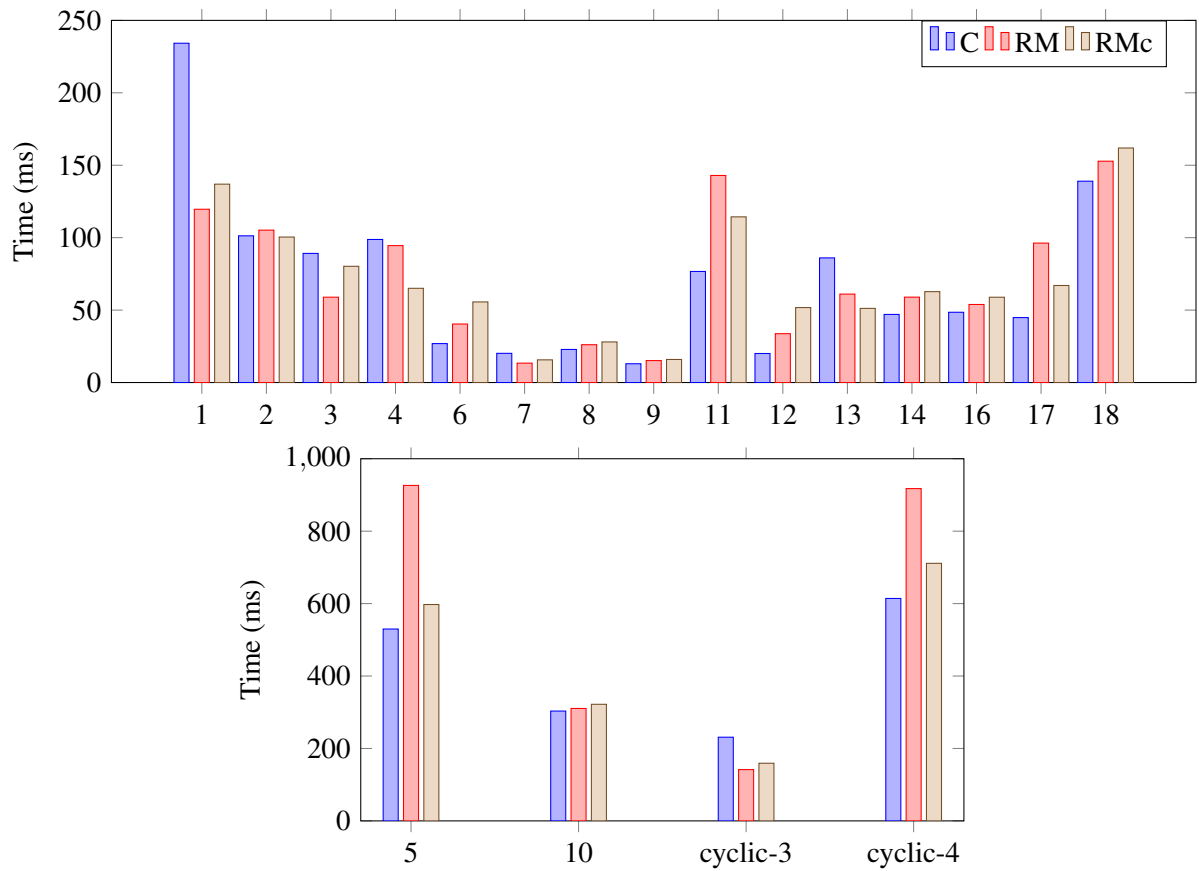

From the experimental viewpoint, the method offers promising results. We can see that in some cases, reduction machines offer some improvements on the classic reduction, while in others they do not. However, the differences are not significant in most cases. Note that so far we only implemented sequential versions. We expect parallel versions to behave more efficiently.

\section{Related Work}
\label{related}
In addition to the contributions mentioned in Subsection~\ref{GBimprovements}, such as Faug\`{e}re's F4, F5, there were other attempts to use different forms of parallelism in order to improve computations of Gr\"obner bases, see for example~\cite{kredel}. These efforts, however, focus on making several reductions in the same time. Reduction machines bring parallelism inside the reduction. In this sense, our work is complementary to such approaches. 

Our implementation is an early prototype, therefore, at this moment we do not compare it to other similar applications such as Java Algebra System (JAS), see~\cite{kredel}, or well established computer algebra systems such as Mathematica, Maple, CoCoA, Magma, Macaulay2, Singular, Sage, etc.

In fact, we tried out the different reduction methods with our prototype implementation in order to have a level playing field for their comparison. Of course, we do not yet have the parallel version of the reduction machines in our implementation, but we expect these to perform much better. A practical way forward is, perhaps, to implement reduction machines in the various existing systems.  

\section{Conclusions and Future Work}
\label{conclusions}

Our main contribution in this paper is the concept of reduction machines and its implementation, together with the proof that they are equivalent to corresponding reduction processes. In fact, we observed the behaviour that led to the idea of reduction machines while trying to prove Theorem~\ref{theorem:sameRemainder}. 

We have a prototype implementation, in Java, of Gr\"obner bases which we used to test the impact of reduction machines on the performance of the Gr\"obner bases algorithm. The results are promising: taking into account that the sequential versions give relatively similar results to the classic reduction. The next obvious step is to implement parallel versions and test them extensively. We are also considering integrating reduction machines in existing open-source systems. 

Also, part of future research is an analysis of the complexity of reduction machines, both in the sequential and in the parallel case. \vspace{7.7 cm}

\pagebreak
\bibliographystyle{eptcs}
\bibliography{bibGroebner}
\pagebreak
\section*{Appendix}
\label{appendix}
\subsubsection*{Collection of problems}
\label{collection}

\begin{enumerate}
	\item $I =\,<x^2+y+z-1, x+y^2+z-1, x+y+z^2-1>$ \cite{cox} 
	\item $I =\,<x^2+y^2+1, x^2  y+2  x  y+x>$ \cite{adams} 
	\item $I =\,<x^2  y-1, x  y^2-x>$ \cite{cox} 
	\item $I =\,<x^2+y^2+z^2-1, x^2+z^2-y, x-z>$ \cite{cox} 
	\item $I =\,<x  z-y^2+z, x^2+y, x  y+1>$ 
	\item $I =\,<x  y-2  y, 2  y^2-x^2>$  \cite{theorists} 
	\item $I =\,<y-x^3, z-x^5>$ 
	\item $I =\,<y  x-x, y^2-x>$ \cite{adams} 
	\item $I =\,<y-x^2, z-x^3>$ \cite{cox} 
	\item $I =\,<x^3  y  z-x  z^2, x  y^2  z-x  y  z, x^2  y^2-z^2>$ \cite{buchberger1985} 
	\item $I =\,<3  x^2  y+2  x  y+y+9  x^2+5  x-3, 2  x^3  y-x  y-y+6  x^3-2  x^2-3  x+3, x^3  y+x^2  y+3  x^3+2  x^2>$ \cite{buchberger1985} 
	\item $I =\,<2  x  y^2+3  x+4  y^2, y^2-2  y-2>$ \cite{adams} 
	\item $I =\,<x^2  y^2+x  y, y^4-y^2>$ 
	\item $I =\,<x^2  y-y+x, x  y^2-x>$ \cite{adams} 
    \item $I = \,<xy-2yz-z, y^2-x^2z+xz, z^2-y^2x+x>$ \cite{theorists} 
	\item $I =\,<x^3+y^2+4  x  y, x  y+1, z^3+2  x^2  y-2  z>$ 
	\item $I =\,<x  y^2-x  z+y, x  y-z^2, x-y  z^4>$ \cite{cox} 
	\item $I =\,<x^4  y^2-z, x^3  y^3-1, x^2  y^4-2  z>$ \cite{cox} 
	\item $I =\,<x+y+z, xy+yz+zx, xyz-1>$ [cyclic-3]
	\item $I =\,<w+x+y+z, wx+xy+yz+zw, wxy+xyz+yzw+zwx, wxyz-1>$ [cyclic-4]
\end{enumerate}
\end{document}